\documentclass[letterpaper, 10 pt, conference]{ieeeconf}  % Comment this line out if you need a4paper

\IEEEoverridecommandlockouts  % This command is only needed if you want to use the \thanks command
\usepackage{amsmath} % assumes amsmath package installed
\usepackage{amssymb}  % assumes amsmath package installed
\usepackage{graphicx}
\usepackage{verbatim}
\usepackage{color}

\newtheorem{thm}{Theorem}
\newtheorem{cor}{Corollary}
\newtheorem{lem}{Lemma} %[section]

\newtheorem{defi}{Definition}
\newtheorem{prop}{Proposition}

\newtheorem{exx}{Example}
\newtheorem{remm}{Remark}

\newenvironment{definition}{\begin{defi}\rm }{\hfill \hspace*{1pt} \hfill $\lrcorner$\end{defi}}
\newenvironment{remark}{\begin{remm}\rm }{\hfill \hspace*{1pt} \hfill $\lrcorner$\end{remm}}
\newenvironment{proposition}{\begin{prop} \rm }{\hfill \hspace*{1pt} \hfill $\lrcorner$\end{prop}}

\newenvironment{example}{\begin{exx}\rm }{\hfill \hspace*{1pt} \hfill $\lrcorner$ \end{exx}}
\newenvironment{proofof}{{\em Proof of }}{\hfill \hspace*{1pt}
\hfill $\blacksquare$}

\newcommand\real{\ensuremath{{\mathbb R}}}
\newcommand\realn{\ensuremath{{\mathbb{R}^n}}}
\newcommand\mymatrix[2]{\left[\begin{array}{#1} #2 \end{array}\right]}
\newcommand{\smallmat}[1]{\left[ \begin{smallmatrix}#1 \end{smallmatrix} \right]}

\newcommand{\calA}{\mathcal{A}}

\newcommand{\calK}{\mathcal{K}}
\newcommand{\calW}{\mathcal{W}}

\title{\LARGE \bf
A dissipativity theorem for $p$-dominant systems
}

\author{Fulvio Forni and  Rodolphe Sepulchre% <-this % stops a space
\thanks{
F. Forni and R. Sepulchre are with the University of Cambridge, Department of Engineering, 
Trumpington Street, Cambridge CB2 1PZ, 
\texttt{f.forni@eng.cam.ac.uk|r.sepulchre@eng.cam.ac.uk}.
The research leading to these results has received funding from the European Research Council under the
Advanced ERC Grant Agreement Switchlet n.670645.} 
}

\begin{document}

\maketitle
\thispagestyle{empty}
\pagestyle{empty}

\begin{abstract}
We revisit the classical dissipativity theorem of linear-quadratic theory in a generalized framework
where the quadratic storage is negative definite in a $p$-dimensional subspace and positive definite in a complementary subspace. 
The classical theory assumes $p=0$ and provides an interconnection theory for stability analysis, i.e.
convergence to a zero dimensional attractor. The generalized theory is shown to provide an interconnection theory for
$p$-dominance analysis, i.e. convergence to a $p$-dimensional dominant subspace. In turn, this property is
the differential characterization of a generalized contraction property for nonlinear systems.
The proposed generalization opens a novel avenue for the analysis of  interconnected systems with low-dimensional attractors.
\end{abstract}

\section{Introduction}

Dissipativity theory \cite{Willems1972} is a cornerstone of system theory. A dissipation inequality
relates the variation of the storage, which relates to an internal system property, to the
supply rate, which expresses how much the environment can affect the internal property.
When the storage is positive definite, the internal property at hand is Lyapunov stability,
and dissipativity theory provides an interconnection theory for the analysis of stability.
Dissipativity theory is constructive for linear systems and quadratic storages, leading to
stability criteria that can be verified through the solution of linear matrix inequalities \cite{Willems1972a}.

In the present paper, we explore the significance of  linear quadratic dissipativity theory
when the  quadratic storage is no longer positive definite but instead has a fixed inertia,
that is, $p$ negative eigenvalues and  $n-p$ positive eigenvalues. We show that the internal
dissipation inequality then characterizes $p$-dominance, that is, the existence of an invariant
$p$-dimensional subspace that attracts all solutions. Dissipativity theory then becomes an
interconnection theory for the analysis of $p$-dominance. In this context, stability can be
interpreted as $0$-dominance, in the sense that a zero-dimensional subspace attracts all solutions.

Our interest in $p$-dominance as a system property stems primarily from its significance in
the differential analysis of nonlinear systems. We use {\it differential} dominance, that is,
infinitesimal dominance along trajectories, to analyse the asymptotic behavior of nonlinear systems with
low-dimensional attractors. Beyond $p=0$, which corresponds to the classical analysis of convergence to
a unique equilibrium, we focus on $p=1$ as a relevant framework to study multistability and on $p=2$
as a relevant framework to study limit cycle oscillations. Hence we primarily think of
linear-quadratic dissipativity theory of $p$-dominance as an interconnection theory for the differential
 analysis of multistable or oscillatory nonlinear systems.

This paper concentrates on the main ideas of the proposed approach, leaving aside many possible generalizations. 
Section \ref{sec:p-dominance} provides the linear-quadratic dissipation characterization
of $p$-dominant linear systems and explains  its link with the contraction of a rank $p$ ellipsoidal cone. Section \ref{sec:p-dissipative} presents a straightforward extension of the fundamental dissipativity theorem to $p$-dominant linear systems. Section \ref{sec:diff-p-dominance} outlines the fundamental property
of $p$-monotone systems and how it generalizes contraction (interpreted as 0-monotonicity)  and monotonicity (interpreted as 1-monotonicity), two system properties that have been extensively studied in nonlinear system theory.  
Section \ref{sec:diff-p-dissipative} extends $p$-dissipativity to the nonlinear setting and provides
a basic illustration of the potential of the theory with 
a simple example of a 2-dominant system that  has a limit cycle oscillation 
resulting from the passive interconnection of two 1-dominant systems.

\section{$p$-dominant linear time-invariant systems}
\label{sec:p-dominance}

\begin{definition}
\label{def:LMI-dominance}
A linear system $\dot{x} = Ax$ is  $p$-dominant with rate $\lambda \ge 0$ 
if and only if there exist a symmetric matrix $P$ with inertia $(p,0,n-p)$  such that
\begin{equation}
\label{eq:LMI-dominance}
A^T P + P A \le  -2 \lambda P + \epsilon I \ . 
\end{equation}
for some $\epsilon \ge 0$. The property is {\it strict} if $\epsilon >0$.
\end{definition}

Equivalent characterizations of $p$-dominance are provided in the following proposition (the proof is in appendix).
\begin{proposition}
\label{prop:equivalences}
For $\epsilon>0$, the Linear Matrix Inequality \eqref{eq:LMI-dominance} is equivalent to any of the following conditions:
\begin{enumerate}
\item The matrix $A+ \lambda I $ has $p$ eigenvalues with strictly positive real part
and $n-p$ eigenvalues with strictly negative real part.
\item there exists an invariant splitting of the vector space $\real^n$ into dominant $E_p$ and
non-dominant $E_{n-p}$ eigenspaces such that any solution of the linear system $\dot x = Ax$ can be written as $x(t)=x_p(t)+x_{n-p}(t)$ 
with $x_p(t) \in E_p$, $x_{n-p}(t) \in E_{n-p}$ and for some $0 < C_p \leq 1 \leq C_{n-p}$ and $\lambda_p < \lambda < \lambda_{n-p}$,
$$
\begin{array}{rcl}
\mid x_p(t)\mid &\geq & C_p \,e^{-\lambda_p t} \mid x_p(0) \mid \vspace{1mm}\\
 \mid x_{n-p}(t) \mid &\leq & C_{n-p}\, e^{-\lambda_{n-p} t} \mid x_{n-p}(0)\mid .
\end{array}
$$
\end{enumerate}
\end{proposition}

The property of $p$-dominance ensures a splitting between $n-p$ {\it transient} modes
and $p$ {\it dominant} modes.  Only the $p$ dominant modes dictate the asymptotic behavior.

\begin{example}
\label{ex:linear_msd}
Consider a simple mass-spring-damper system
with mass $m =1$, elastic constant $k = 1$  and damping coefficient $c = 4$,
that is,
$
\dot{x} =  \smallmat{0 & 1 \\ -1 & -4} x\ .
$
Eigenvalues are in $-0.2679$ and $-3.7321$. 
The system is $p$-dominant with
$
P =  \smallmat{ -0.4338  & 0.6535 \\ 0.6535 &   1.4338}     
$
solution to \eqref{eq:LMI-dominance} with $\lambda = 0.2679 + 1$ computed
using Yalmip \cite{Yalmip2004}. 
$P$ has inertia $(1,0,1)$.
Following Proposition \ref{prop:positivity}, the system is positive
with respect to the cone represented in Figure \ref{fig:cone_msd}, left.

Note that the eigenvector of $P$ related to the negative eigenvalue belongs
to the interior of the cone $\{x\in\realn\,|\, x^T P x \leq 0\}$, closer to the position axis. The other eigenvector 
characterizes the direction transversal to the cone, closer to the velocity axis, 
as expected. The tendency of these two eigenvectors to align with position and velocity axes becomes clear for large damping values.
For example, for 
$
\dot{x} = \smallmat{0 & 1 \\ -1 & -8} x
$
we get
$
P =  \smallmat{  -0.9193  &  0.2177 \\  0.2177  &  1.9193}
$,
whose cone is in Figure \ref{fig:cone_msd}, right.
\begin{figure}[htbp]
\centering
\includegraphics[width=0.49\columnwidth]{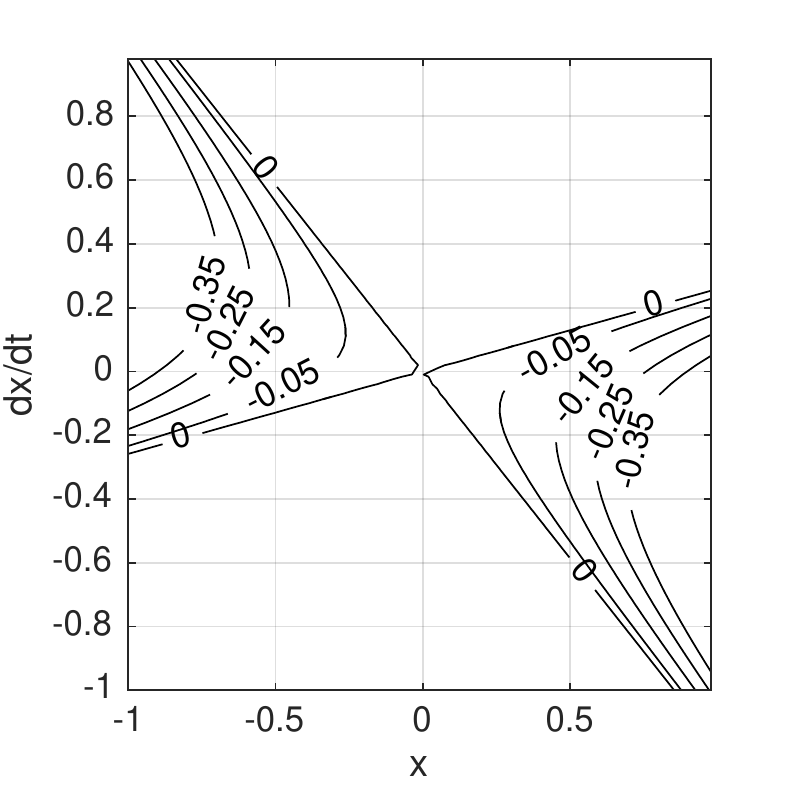} 
\includegraphics[width=0.49\columnwidth]{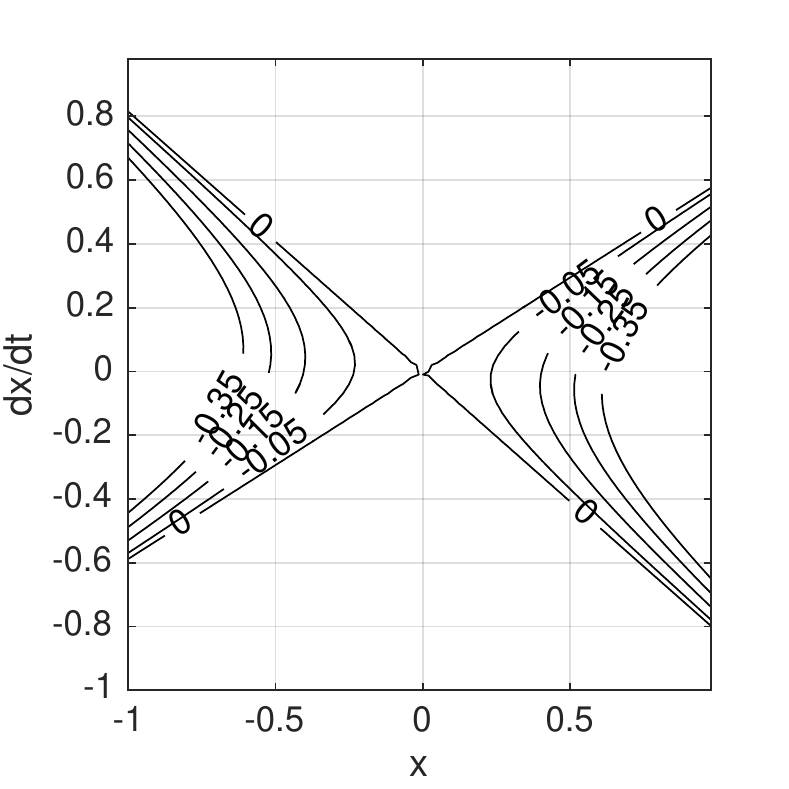}
\vspace*{-7mm}
\caption{Cones for the mass-spring-damper system 
for $c=4$ (left) and $c=8$ (right).}
\label{fig:cone_msd}
\end{figure}
\end{example}

$p$-dominance is a generalization of the classical property of exponential stability, which corresponds to $p=0$: all modes are transient and the asymptotic behavior
is $0$-dimensional. 
In contrast, for $p>0$, the property of $p$-dominance is closely connected to the the classical property of positivity
 \cite{Bushell1973,Luenberger1979,Vandergraft1968}.
We recall that a linear system $\dot{x} = Ax$ 
is positive with respect to a cone $\calK \subseteq \realn$ if
$e^{A t} \calK \subseteq \calK$ for all $t \geq 0$. Strict positivity
further requires that the system maps any nonzero vector of the boundary
of $\calK$ into the interior of the cone, for $t>0$.

\begin{proposition}
\label{prop:positivity}
For any $0< p< n$, 
a $p$-dominant system is strictly positive with respect to the cone
\begin{equation}
\calK := \{x \in \realn \,|\, x^T P x \leq 0 \}
\end{equation} 
where $P$ is any solution to \eqref{eq:LMI-dominance}.
\end{proposition}
\begin{proof}
We have to prove that any vector on the boundary 
of the cone $V(x) := x^TPx = 0$ is mapped in the cone. 
Clearly,
$\dot{V}(x) = x^T (A^T P + P A) x < -2\lambda x^T P x = 0$,
which shows the invariance of $\calK$. It also shows that any
nonzero vector on the boundary is mapped into the interior of $\calK$,
for $t > 0$.
\end{proof}

For $p=1$, \eqref{eq:LMI-dominance} expresses the contraction of an ellipsoidal cone.
The $1$-dimensional dominant subspace is spanned by 
the Perron-Frobenius eigenvector.
For $p>1$, $p$-dominance is also a positivity property with 
respect to higher order cones 
\cite{Fusco1991,Sanchez2009,Sanchez2010,Mostajeran2017}.

The reader will notice that there is an important distinction between positivity in the sense of 
Proposition \ref{prop:positivity} and dominance in the sense of Definition \ref{def:LMI-dominance}. 
For $p>0$, any dominant system is strictly positive but the converse is not true. This is due to the extra requirement
of a {\it nonnegative} dissipation rate $\lambda \ge 0$ in the definition of $p$-dominance. The difference is significant
because if affects the type of contraction associated to each property. From Proposition \ref{prop:equivalences},  $p$-dominance is a form
 of {\it horizontal} contraction in the sense of \cite{Forni2014}: the vector space is splitted into a vertical space of dimension $p$
and a horizontal space of dimension $n-p$; contraction is imposed  in the horizontal space only. This property requires a nonnegative
dissipation rate $\lambda \ge 0$. In contrast, positivity is only a form of {\it projective} contraction, which does not require a nonnegative dissipation
rate. For $p=1$ the projective contraction is captured by the contraction
of the Hilbert metric \cite{Bushell1973,Pratt1982,Bonnabel2011}.
The next proposition provides
a projective contraction measure for a general $p$.

 \begin{proposition}
For any given $p$-dominant system of dimension $n$, 
there exist positive semidefinite matrices $P_u$ and $P_s$
of rank $p$ and $n-p$ respectively such that,
given $|x|_u := \sqrt{x^T P_u x}$ and $|x|_s := \sqrt{x^T P_s x}$,
the ratio $ |x(t)|_s / |x(t)|_u $ is 
exponentially decreasing
along any trajectory $x(\cdot)$ of the system from $|x(0)|_u \neq 0$.
\end{proposition}
\begin{proof}
By Proposition \ref{prop:equivalences}, $A+\lambda I$ has $p$ unstable eigenvalues
and $n-p$ stable eigenvalues. 
Thus, 
there exist matrices $P_u$ and $P_s$ of rank $p$ and $n-p$ respectively,
and a small $\varepsilon > 0$,  such that
$
A^T P_u + P_u A  \geq (-2 \lambda + \varepsilon) P_u 
$
and 
$
A^T P_s + P_s A \leq (-2 \lambda- \varepsilon) P_s 
$ 
Define $U(x) = x^T P_u x$ and $S(x) = x^T P_s x$. Then,
by comparison theorem, along any trajectory of the system we have
$
U(x(t)) \geq e^{(-2 \lambda + \varepsilon)t} U(x(0))
$
and
$
S(x(t)) \leq e^{(-2 \lambda - \varepsilon)t} S(x(0))
$.
It follows that
\begin{equation}
S(x(t)) / U(x(t)) \leq e^{-2 \varepsilon t} S(x(0)) / U(x0))
\end{equation}
which guarantees that the ratio $ |x(t)|_s / |x(t)|_u $ is strictly decreasing
and converges to zero as $t \to \infty$. 
\end{proof}

\section{$p$-dissipativity}
\label{sec:p-dissipative}
The internal property of $p$-dominance is captured by the
matrix inequality \eqref{eq:LMI-dominance} which enforces a
conic constraint between the state of the system 
and its derivative of the form
\begin{equation}
\label{eq:internal}
\mymatrix{c}{\!\!\dot x \!\! \\ \!\! x \!\!}^T \!
\mymatrix{cc}{
0 & P \\ P & 2\lambda P + \varepsilon I
}
\mymatrix{c}{\!\!\dot x \!\! \\ \!\! x \!\!}
\leq 0
\end{equation}
where $P$ is a matrix with inertia $(p,0,n-p)$
and $\varepsilon > 0$. 
A system is  $p$-dominant  if  the linear relationship between
$\dot{x}$ and $x$ satisfy \eqref{eq:internal}.  

Dissipativity theory  extends $p$-dominance to {\it open} systems
by augmenting the internal dissipation inequality with an external supply.
The external property of  \emph{$p$-dissipativity} is captured by a
conic constraint between the state of the system $x$, its derivative $\dot{x}$,
and the external variables $y$ and $u$ of the form
\begin{equation}
\label{eq:p-dissipativity}
\mymatrix{c}{\!\!\!\dot x\!\!\! \\ \!\!\!x\!\!\!}^T \!\!
\mymatrix{cc}{
0 & P \\ P & 2\lambda P + \varepsilon I
}\!
\mymatrix{c}{\!\!\!\dot x\!\!\! \\ \!\!\!x\!\!\!}
\leq
\mymatrix{c}{\!\!\!y\!\!\! \\ \!\!\!u\!\!\!}^T \!\!
\mymatrix{cc}{
\!Q\! & \!L\! \\ \!L^T\! & \!R\!
}\!
\mymatrix{c}{\!\!\!y\!\!\! \\ \!\!\!u\!\!\!} 
\end{equation}
where $P$ is a matrix with inertia $(p,0,n-p)$, $\lambda \ge 0 $,
$L, Q,R$ are matrices of suitable dimension,
and $\varepsilon \geq 0$. The property is {\it strict} if $\epsilon >0$.
We call supply rate
$s(y,u) := y^T Q y + y^TL u + u^T L^T y + u^T R u$ the right-hand side of \eqref{eq:p-dissipativity}.
An open dynamical system  
is \emph{$p$-dissipative with rate $\lambda$} if its
dynamics $\dot{x} = Ax + Bu$, $y=Cx + Du$ satisfy \eqref{eq:p-dissipativity}
for all $x$ and $u$.

$p$-dissipativity has a simple characterization in terms of 
matrix inequalities.
\begin{proposition}
\label{prop:LMI-dissipavitity}
A linear system $\dot{x} = Ax + Bu$, $y=Cx + Du$ is $p$-dissipative with rate $\lambda$
if and only if there exist a 
symmetric matrix $P$ with inertia $(p,0,n-p)$ such that
\begin{equation}
\label{eq:LMI-dissipativity}
\smallmat
{\!
A^T\!  P + P A + 2\lambda P -C^T\! Q C + \varepsilon I \!&\! P B - C^T\!L - C^T\! Q D \! \\
 \! B^T P - L^T C - D^T Q C \!&\! -D^T Q D - L^TD - D^T L - R \!
 }
 \leq 0\, .
\end{equation}
\end{proposition}
\begin{proof}
[$\Rightarrow$]
Just replace $\dot{x} = Ax+Bu$ and $y = Cx+Du$ in \eqref{eq:p-dissipativity}
and rearrange. [$\Leftarrow$] Multiply \eqref{eq:LMI-dissipativity} by
$[x^T \, u^T]$ on the left, and by $[x^T \, u^T ]^T$ on the right. Then we get
$
\dot{x}^T P x + x^T P \dot{x} +2\lambda x^T Px <  s(y,u)
$
 as desired.
\end{proof}

An interconnection theorem can be easily derived.
\begin{proposition}
\label{thm:dissipativity}
Let $\Sigma_1$ and $\Sigma_2$ 
$p_1$-dissipative and $p_2$-dissipative
systems respectively, with
uniform rate $\lambda$ and 
with supply rate
\begin{equation}
s_i(y_i,u_i) = \mymatrix{c}{\!\!y\!\! \\ \!\!u\!\!}^T \!
\mymatrix{cc}{
Q_i & L_i \\ L_i^T & R_i
}
\mymatrix{c}{\!\!y\!\! \\ \!\!u\!\!} 
\end{equation}
for $i \in \{1,2\}$. The closed-loop system given 
by negative feedback interconnection
\begin{equation}
\label{eq:feedback_interconnection}
u_1 = -y_2 + v_1 \qquad u_2 = y_1 + v_2
\end{equation}
is $(p_1+p_2)$-dissipative with rate $\lambda$
from $v = (v_1,v_2)$ to $y = (y_1,y_2)$
with supply rate 
\begin{equation}
\label{eq:cl_supply}
s(y,v) = 
\mymatrix{c}{\!\!\!y \!\!\! \\ \!\!\!v\!\!\!}^{\!T} \!\!
{\footnotesize 
\mymatrix{cc|cc}{
Q_1 + R_2 & -L_1 + L_2^T & L_1 & R_2 \\
-L_1^T + L_2 & Q_2  + R_1 & -R_1 & L_2 \\ \hline
L_1^T &  -R_1 & R_1 & 0 \\
R_2 & L_2^T & 0 & R_2 
}\!
}\!
\mymatrix{c}{\!\!\!y\!\!\! \\ \!\!\!v\!\!\!}  .
\end{equation}
The closed loop is $(p_1+p_2)$-dominant with rate $\lambda$ if 
\begin{equation}
\label{eq:dissipativity_interconnection}
\mymatrix{cc}{
Q_1 + R_2 & -L_1 + L_2^T \\ -L_1^T + L_2 & Q_2  + R_1
} 
\leq 0 \ . \vspace{-3mm}
\end{equation}
\end{proposition}
\begin{proof}
By standard steps on interconnection of dissipative systems.
Define $P := \smallmat{P_1 & 0 \\ 0 & P_2}$ which, by construction,
has inertia $(p_1+p_2, 0, n_1+n_2 - p_1-p_2)$,
where $n_i$ is the dimension of the state of system $i\in\{1,2\}$.
Then, a simple computation shows that 
the closed-loop system given by \eqref{eq:feedback_interconnection},
satisfies \eqref{eq:p-dissipativity} from $v$ to $y$
with supply rate \eqref{eq:cl_supply}.
Furthermore, denote by $S$ the left-hand side of
\eqref{eq:dissipativity_interconnection} 
and take $v = 0$. Then, using the aggregate state $x=(x_1,x_2)$,
\eqref{eq:dissipativity_interconnection} guarantees 
$p$-monotonicity since
$$
\mymatrix{c}{\!\!\!\dot x\!\!\! \\ \!\!\!x\!\!\!}^T \!\!\!
\smallmat{
0 & P \\ P & \, 2\lambda P + \varepsilon I
}\!
\mymatrix{c}{\!\!\!\dot x\!\!\! \\ \!\!\!x\!\!\!}
\leq
\mymatrix{c}{\!\!\!y_1 \!\!\! \\ \!\!\!y_2\!\!\!}^T \!\!\!
S
\mymatrix{c}{\!\!\!y_1\!\!\! \\ \!\!\!y_2\!\!\!} \leq 0 \ .
\vspace{-5mm}
$$
\end{proof}

For $p_1 = p_2 = 0$ and $\lambda=0$ the
proposition reduces to the standard interconnection
theorem for dissipative systems, \cite{Willems1972a}.
Like classical dissipativity is related
to the internal property of stability, $p$-dissipativity
is tightly related to the internal property of 
$p$-dominance. Thus, for uniform rate $\lambda$, 
Proposition \ref{thm:dissipativity} provides interconnection conditions that lead
to $p$-dominant closed loop systems. 

\begin{example}
\label{ex:linear_msd2}
We say that
a system is \emph{$p$-passive with rate $\lambda$}
when the dissipation inequality \eqref{eq:p-dissipativity}
is satisfied with supply rate 
\begin{equation}
\label{eq:passive}
s(y,u) = \mymatrix{c}{\!\!y\!\! \\ \!\!u\!\!}^T \!
\mymatrix{cc}{
0 & I \\ I & 0
}
\mymatrix{c}{\!\!y\!\! \\ \!\!u\!\!}  \ .
\end{equation}
For $D=0$, taking $\varepsilon = 0$ for simplicity,
\eqref{eq:LMI-dissipativity} reduces to 
\begin{equation}
\label{eq:passivity}
\left\{
\begin{array}{rcl}
A^T  P + P A   &\leq& - 2\lambda P \\
PB &=& C^T \ .
\end{array}
\right.
\end{equation}
For the open mass-spring system
$
\dot{x} = \smallmat{0 & 1 \\ -1 & -8} x + \smallmat{0 \\ 1} u
$, 
$y = \smallmat{0 & 1} x$, the matrix
$
P =  \smallmat{ -1  &  0 \\  0  &  1}
$
satisfies \eqref{eq:passivity}
with $\lambda = 1.2679$, which gives $1$-passivity from $u$ to $y$.
Proposition \ref{thm:dissipativity}
guarantees that the negative feedback interconnection $u = - k y+v$ 
is $1$-passive system  with rate $\lambda$, for any gain $k \geq 0$.
Indeed, the closed-loop system is
$1$-dominant for any $k\geq 0$.

A generalized small gain theorem can also be illustrated.
We say that a system has \emph{$p$-gain $\gamma$ from $u$ to $y$ 
with rate $\lambda$}
when the dissipation inequality \eqref{eq:p-dissipativity}
is satisfied with supply rate 
\begin{equation}
\label{eq:L2}
s(y,u) = \mymatrix{c}{\!\!y\!\! \\ \!\!u\!\!}^T \!
\mymatrix{cc}{
-I & 0 \\ 0 & \gamma^2 I
}
\mymatrix{c}{\!\!y\!\! \\ \!\!u\!\!}  \ .
\end{equation}
For the open mass-spring system,
the matrix $P$ above satisfies \eqref{eq:p-dissipativity}
with $L = -I$, $Q=0$ and $R = \gamma^2 I$ for $\gamma = 0.3$.
Thus, Proposition \ref{thm:dissipativity} guarantees that 
the closed loop given by the feedback interconnection $u=\pm k y$
is $1$-dominant for any $-3.3 < k < 3.3$.

The resulting internal dominance is the analog of the resulting stability
in classical applications of the passivity theorem and of the small gain theorem
\cite{Sepulchre1997,VanDerSchaft1999}.
\end{example}

\section{Differential analysis and $p$-dominance}
\label{sec:diff-p-dominance}
$p$-dominance and $p$-dissipativity are not limited to linear systems.
In nonlinear analysis 
it makes sense to study these properties  infinitesimally or differentially, see e.g.
\cite{Forni_Sepulchre_tutorial_cdc2014}.
In what follows we will see how $p$-dominance of the system linearization restricts
the asymptotic behavior of the nonlinear system. 

The nonlinear system $\dot{x} = f (x)$ ($f$ smooth, 
$x\in \realn$) is (differentially)  \emph{$p$-dominant with rate $\lambda \ge 0$} and constant storage $P=P^T$ if
the prolonged system \cite{Crouch1987}
\begin{equation}
\label{eq:prolonged}
\dot{x} = f(x) \qquad \dot{\delta x} = \partial f(x) \delta x
\end{equation}
satisfies the conic constraint
\begin{equation}
\label{eq:diff-internal}
\mymatrix{c}{\!\!\dot {\delta x}\!\! \\ \!\!\delta x\!\!}^T \!
\mymatrix{cc}{
0 & P \\ P & 2\lambda P + \varepsilon I
}
\mymatrix{c}{\!\!\dot {\delta x}\!\! \\ \!\!\delta x\!\!}
\leq 0
\end{equation}
for every $\delta x \in \realn$, 
where $P$ is a matrix with inertia $(p,0,n-p)$
and  $\varepsilon \ge 0$. The property is strict if $\epsilon >0$.
From \eqref{eq:diff-internal}, the reader will recognize that 
differential dominance is just dominance of the linearized dynamics.
In this paper, we only consider the case of a {\it constant} matrix $P$, but generalizations
might be considered.

Differential dominance for $p=0$ is another synonym of   differential stability, or  contraction, or convergence 
\cite{Lohmiller1998,Pavlov2005,Russo2010,Forni2014}.
The trajectories of the nonlinear system converge {\it exponentially} towards each other.
In fact, \eqref{eq:diff-internal} reduces to the inequality
\begin{equation}
\label{eq:contraction}
\partial f(x)^T P + P \partial f(x) \le -\epsilon I \qquad \forall x\in \realn
\end{equation}
where $P$ is a positive definite matrix. This is a 
typical condition for contraction with respect to a  constant (Riemannian) metric
$P$ \cite{Forni2014}. The following proposition is a 
straightforward consequence of \eqref{eq:contraction}.

\begin{proposition}
\label{thm:contraction} If $\dot x=f(x)$ is strictly $0$-dominant, then all solutions
exponentially converge to a unique fixed point.
\end{proposition}

\begin{proof}
Denote by $\psi^t(x)$ the semiflow of $\dot{x} = f(x)$,
characterizing the trajectory of the system at time $t$ from the initial condition 
$x$ at initial time zero. Then, 
contraction guarantees that 
the distance $d_P(x_1,x_2) :=  ((x_1-x_2)^TP(x_1-x_2))^{1/2}$ 
is exponentially decreasing along any pair of trajectories $x_1(\cdot),x_2(\cdot)$, \cite{Pavlov2005,Forni2014}. 
Indeed, for any $t >0$, $\psi^t(\cdot)$ is a contraction mapping on $\realn$
with respect to $d_P$.
By Banach fixed point theorem
$\psi^t(\cdot)$ admits a unique fixed-point $x^*$ in $\realn$.
By contraction every trajectory converges to $x^*$ exponentially. 
\end{proof}
 
 For $p=1$, differential dominance is closely related to 
differential positivity with respect to a constant 
ellipsoidal cone $\calK\subseteq \realn$, \cite{Forni2016}.
A straightforward adaptation of 
Proposition \ref{prop:positivity} shows that 
the linearized trajectories of a differentially $1$-dominant system
map the boundary of the cone
$\calK := \{\delta x\in \realn\,|\, \delta x^T P \delta x \leq 0\}$
into its interior, as required by strict differential positivity.
We observe that
$\calK$ is the union of two pointed convex cones
$\calK = \calK_{+} \cup \calK_{-}$ such that
$\calK_{+} \cap \calK_{-} = \{0\}$\footnote{Consider 
any hyperplane $\calW \in \realn$ such that $\calW\cap \calK = \{0\}$.
Let $w$ be the normal vector to $\calW$
then
$\calK_{+} := \{\delta x\in\calK\,|\, w^T \delta x \geq 0 \}$ and
$\calK_{-} := \{\delta x\in\calK\,|\, w^T \delta x \leq 0 \}$ are solid, pointed and convex
cones.}.
Then,
differential 1-dominance also guarantees strict
differential positivity with respect to $\calK_{+}$
(and to $\calK_{-}$), which follows from 
strict differential positivity with respect to $\calK$ and from the fact that
the contact point between $\calK_{+}$ and $\calK_{-}$ is $\delta x= 0$,
which is a fixed point of the linearization.
Thus, differential $1$-dominance guarantees strict
differential positivity with respect to a constant
solid, pointed, convex cone in $\realn$, 
which leads to the following result.

\begin{proposition}
\label{thm:almost_convergence}
If $\dot x=f(x)$ is strictly $1$-dominant, then all bounded solutions
exponentially converge to some fixed point.
\end{proposition}

\begin{proof}
The details of the proof are omitted but the the proof closely follows the arguments in  \cite[Corollary 5]{Forni2016}. 
\end{proof}

Differential $1$-dominance also closely relates to monotonicity 
\cite{Smith1995,Hirsch2003,Angeli2003,Hirsch2006}
with respect to the partial order $\preceq$ 
induced by $\calK_{+}$: $x \preceq y$ iff $y-x \in \calK_{+}$.
In fact, differential $1$-dominance guarantees that
any pair of trajectories $x_1(\cdot)$, $x_2(\cdot)$ of the
nonlinear system from ordered initial conditions
$x_1(0) \preceq x_2(0)$ satisfy
$x_1(t) \preceq x_2(t)$ for all $t\geq 0$,
as required by classical monotonicity
(a direct consequence of the relation with differential positivity \cite{Forni2016}).
In this sense, Proposition \ref{thm:almost_convergence} 
is the counterpart of  the well know property that
almost every bounded trajectory of a monotone system 
converges to a fixed point \cite{Hirsch1988,Smith1995}. The fact that the property
holds for {\it all } trajectories in differentially 1-dominant systems follows from
the nonnegative dissipation rate $\lambda \ge 0$, which is not necessary for monotonicity.

For $p=2$, differential dominance closely relates to the notion of 
monotonicity with respect rank-$2$ cones of \cite{Sanchez2009,
Smith1986}, see
\cite[Equations (7) and (8)]{Sanchez2009}. %\begin{proposition}
\begin{proposition}
\label{thm:limit_cycle}
If $\dot x=f(x)$ is strictly $2$-dominant, then all bounded solutions whose $\omega$-limit set does not
contain an equilibrium point  exponentially converge to a closed orbit.
\end{proposition}

\begin{proof}
The details of the proof are left to an extended version of this paper
but the the proof closely follows the arguments in  \cite[Theorem 1]{Sanchez2009}. 
\end{proof}

Proposition \ref{thm:limit_cycle} shows that
differentially $2$-dominant systems enjoy properties akin to the 
Poincare-Bendixson theory of planar systems.

\section{Differential analysis and $p$-dissipativity}
\label{sec:diff-p-dissipative}

In analogy with the previous section, we use
the prolonged system to define $p$-dissipativity in a  nonlinear setting.
For simplicity we will consider systems of the form 
$\dot{x} = f(x) + Bu$, $y=Cx$ 
($f$ is smooth, $x\in \realn$, $u\in \real^m$),
whose prolonged system \cite{Crouch1987} is given by 
\begin{equation}
\label{eq:open-prolonged}
\left\{
\begin{array}{rcl}
\dot{x} &=& f(x)+Bu \\
y &=& C x
\end{array}
\right.
\qquad 
\left\{
\begin{array}{rcl}
\dot{\delta x} &=& \partial f(x) \delta x + B \delta u \\
\delta y &=& C \delta x \ .
\end{array}
\right.
\end{equation}

A nonlinear system 
is \emph{differentially $p$-dissipative with rate $\lambda \ge 0$}  if
its prolonged system satisfies the conic constraint
\begin{equation}
\label{eq:diff-p-dissipativity}
\mymatrix{c}{\!\!\!\dot{\delta x}\!\!\! \\ \!\!\!\delta x\!\!\!}^T \!\!
\mymatrix{cc}{
\!0\! & \!P\! \\ \!P\! & \!2\lambda P \!+\! \varepsilon I\!
}\!
\mymatrix{c}{\!\!\!\dot{\delta x}\!\!\! \\ \!\!\!\delta x\!\!\!}
\leq
\mymatrix{c}{\! \!\!\delta y \!\!\! \\ \!\!\! \delta u\!\!\!}^T \!\!
\mymatrix{cc}{
\!Q\! & \! L \! \\ \! L^T \! & \! R \!
}\!
\mymatrix{c}{\!\!\!\delta y\!\!\! \\ \!\!\! \delta u\!\!\!} 
\end{equation}
for every $\delta x\in \realn$ and every $\delta u \in \real^m$,
where $P$ is a matrix with inertia $(p,0,n-p)$,
$L, Q,R$ are matrices of suitable dimension,
and $\varepsilon > 0$.

The following interconnection result easily follows.
\begin{proposition}
\label{thm:diff_dissipativity}
Let $\Sigma_1$ and $\Sigma_2$ 
differentially $p_1$-dissipative and differentially $p_2$-dissipative
respectively, with
uniform rate $\lambda$ and 
with differential supply rate
\begin{equation}
s_i(\delta y_i,\delta u_i) = \mymatrix{c}{\!\!\delta y\!\! \\ \!\!\delta u\!\!}^T \!
\mymatrix{cc}{
Q_i & L_i \\ L_i^T & R_i
}
 \mymatrix{c}{\!\!\delta y\!\! \\ \!\!\delta u\!\!}
\end{equation}
for $i \in \{1,2\}$.
The closed-loop system given 
by negative feedback interconnection \eqref{eq:feedback_interconnection}
is differentially $(p_1+p_2)$-dissipative with rate $\lambda$
from $v = (v_1,v_2)$ to $y = (y_1,y_2)$
with differential supply rate 
with supply rate 
\begin{equation}
\label{eq:cl_diff_supply}
s(y,v) \!=\!
\mymatrix{c}{\!\!\!\delta y \!\!\! \\ \!\!\!\delta v\!\!\!}^{\!T} \!\!
{\footnotesize 
\mymatrix{cc|cc}{
Q_1 + R_2 & -L_1 + L_2^T & L_1 & R_2 \\
-L_1^T + L_2 & Q_2  + R_1 & -R_1 & L_2 \\ \hline
L_1^T &  -R_1 & R_1 & 0 \\
R_2 & L_2^T & 0 & R_2 
}\!
}\!
\mymatrix{c}{\!\!\!\delta y\!\!\! \\ \!\!\!\delta v\!\!\!}  .
\end{equation}
The closed loop is differentially $(p_1+p_2)$-dominant if \eqref{eq:dissipativity_interconnection} holds.
\end{proposition}

Proposition \ref{thm:diff_dissipativity} provides an  interconnection
result  for  differential $p$-dominance.
For example, under the assumptions of the proposition, 
the closed loop of a $1$-dominant system (typically a monotone system) 
with a $0$-dominant system (typically a contractive system)
is necessarily $1$-dominant.
In a similar way, the closed loop of two $1$-dominant 
systems leads to $2$-dominance. In this sense,  Proposition \ref{thm:diff_dissipativity} 
provides an interconnection mechanism
to generate periodic behavior ($2$-dominance) from 
the interconnection of multi-stable components ($1$-dominance).

\begin{remark}
Both $p$-dominance and $p$-dissipativity
can be tested algorithmically via simple relaxations.
From \eqref{eq:prolonged} and \eqref{eq:diff-internal} $p$-dominance requires
that 
\begin{equation}
\label{eq:p-monotonicityLMI}
\partial f(x)^T P + P \partial f(x) < -2\lambda P \qquad \forall x\in \realn \ .
\end{equation}
From \eqref{eq:open-prolonged} and 
\eqref{eq:diff-p-dissipativity} differential $p$-dissipativity requires
that
\begin{equation}
\label{eq:differential-p-dissipativityLMI}
{\small\mymatrix{cc}
{\!\!
\!\partial\! f(x)^T \! P \!+\! P \partial\! f(x) \!+ \!2\lambda P \!-\!C^T  \! Q C \!+\! \varepsilon I \! & \! P B \!-\! C^T\!L  \!\!\! \\
 \!\! B^T \!P \!-\! L^T C \!\!  & \!\! - R \!\!\!
 }}
\!\leq\! 0 \, .
\end{equation}

Let $\calA := \{A_1, \dots, A_N\}$ be a family of matrices such that
$\partial f(x) \in \mathit{ConvexHull}(\calA)$ for all $x$. Then,
by construction, 
any (uniform) solution $P$ to 
\begin{equation}
\label{eq:p-monotonicityLMIrelax}
A_i^T P + P A_i < -2\lambda P  \qquad 1\leq i\leq N
\end{equation}
is a solution to \eqref{eq:p-monotonicityLMI}. 
Also, any (uniform) solution $P$ to 
\begin{equation}
\label{eq:differential-p-dissipativityLMIrelax}
{\small 
\mymatrix{cc}
{\!
\!\!A_i^T  P + P  A_i + 2\lambda P -C^T Q C + \varepsilon I \!\!& \!\!P B - C^TL  \!\!\! \\
 \!\! B^T P - L^T C\!\!  &\!\! - R \!\!\!
 }}
\leq  0  
\end{equation}
for $1\leq i\leq N$
is a solution to \eqref{eq:differential-p-dissipativityLMI}. 

We recall that if $\lambda$ is chosen in such a way that each
$A_i+\lambda I$ has exactly $p$ unstable eigenvalues, 
then $P$ necessarily have inertia $(p,0,n-p)$.
\end{remark}

\begin{example}
\label{ex:nonlinear_msd}
We revisit Examples \ref{ex:linear_msd} and \ref{ex:linear_msd2} 
by replacing the linear spring 
with a nonlinear active component $\phi(x_1)$
\begin{equation}
\label{eq:nonlinear_osc}
\left\{
\begin{array}{rcl}
\dot{x_1} &=& x_2  \\
\dot{x_2} &=& \phi(x_1) - 8x_2 + u \ .\\
\end{array}
\right.
\end{equation}
A strictly monotone nonlinear spring 
$-2 \leq \partial \phi(x_1) < -1/2$ makes 
the system $0$-dominant with rate $\lambda = 0$. For instance, 
the linearization reads $\dot{\delta x} = A(x) \delta x + B u$  with
\begin{equation}
A(x) := \mymatrix{cc}{0 & 1\\ \partial \phi(x_1) & -8}  \,,\
B := \mymatrix{c}{0 \\ 1} \ ,
\end{equation}
and the matrix $P := \smallmat{  1  &  0.5 \\  0.5  &  1 }$ 
satisfies the inequality $A(x)^T P + PA(x)< - \lambda P$ 
uniformly in $x$ for $\lambda = 0$. By Proposition \ref{thm:contraction} 
the trajectories of the system exponentially converge to the unique fixed point
of the system.

A non-monotone spring $-3 \leq \partial \phi(x_1) < 1$ 
shifts the eigenvalues of $A(x)$  to the right-half complex plane.
$0$-dominance can no longer hold. However, the  matrix
$P := \smallmat{-1 & 0 \\ 0 & 1}$ 
satisfies $A(x)^T P + PA(x)< - \lambda P$ uniformly in $x$
for $\lambda = 1$, which makes the system $1$-dominant.
The system is differentially $1$-passive with rate
$\lambda$  from $u$ to $y=Cx = x_2$ 
since $[\,0 \ 1\,]^T =: C^T \!= PB$. 
Propositions \ref{thm:almost_convergence} and \ref{thm:diff_dissipativity}
guarantee that the closed loop system given by $u=-k y$ is 
$1$-dominant for any $k\geq 0$. All trajectories 
of the closed loop converge to a fixed point.

Different outputs can also be considered.
For example, again for $-3 \leq \partial \phi(x_1) < 1$,
the dissipation inequality is satisfied
uniformly in $x$ by the matrix $P := \smallmat{  -2  &  1 \\  1  &  2 }$ 
which makes \eqref{eq:nonlinear_osc} differentially $1$-passive from 
$u$ to $y=Cx = x_1 + 2x_2$ since $[\,1 \ 2\,]^T =: C^T \!= PB$.
By Proposition \ref{thm:diff_dissipativity}, the negative feedback
interconnection of two mass-spring-damper systems is 
differentially $2$-passive with rate $\lambda = 1$.
The $4$-dimensional system is $2$-dominant and, with some extra work, it can be
proven that all solutions converge either to the unstable equilibrium or to a limit cycle,
as illustrated by the simulation in Figure \ref{fig:osc},
where $\phi(x_1)= x_1 - \frac{1}{3}\min(x_1^2,4) x_1$. 
\begin{figure}[htbp]
\centering
\includegraphics[width=0.49\columnwidth]{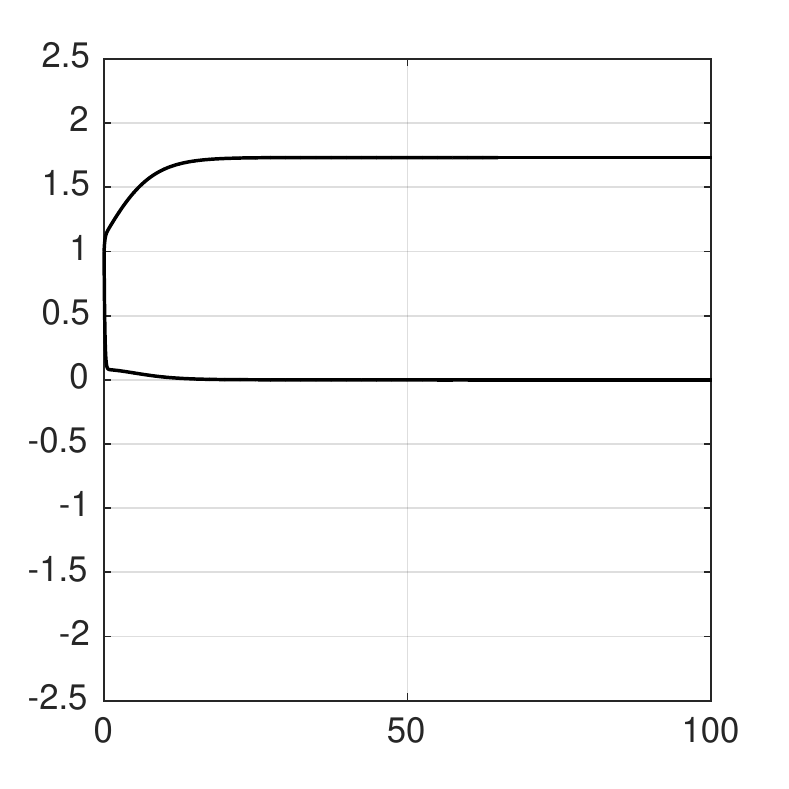} 
\includegraphics[width=0.49\columnwidth]{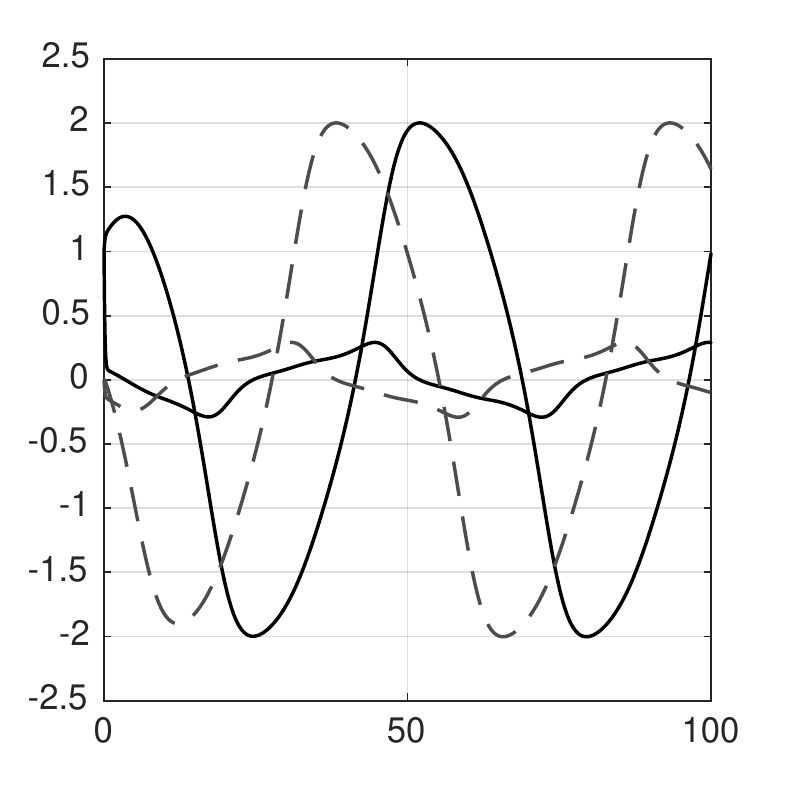}
\vspace*{-7mm}
\caption{\textbf{Left:} 
position and velocity of the single nonlinear mass spring damper system with
nonlinear spring $\phi(x_1)= x_1 - \frac{1}{3}\min(x_1^2,4) x_1$. 
Unforced behavior (stable fixed point)
from the initial condition $x_1 = 1$, $x_2 = 1$.
\textbf{Right:} positions variables of the closed loop of two mass spring damper systems 
with nonlinear spring $\phi(x_1)$ in feedback interconnection
from the input $u$ to the output $y=x_1 + 2x_2$. 
Trajectories converge to a limit cycle.}
\label{fig:osc}
\end{figure}
\end{example}

\section{Conclusions}
We introduced the notions of $p$-dominance and $p$-dissipativity
both in the linear and nonlinear settings. They provide a 
conceptual and algorithmic framework for 
the analysis of multi-stable and periodic behaviors. Interconnection
theorems are provided which extend classical dissipativity
theory with indefinite storages. The example illustrates the potential
of the approach.

The paper  only exposes the basic ideas of the proposed approach. 
Future research directions will include
a frequency domain characterization for $p$-dominance and $p$-dissipativity.
In the nonlinear setting it is relevant to extend the framework to
the case of non-constant matrix $P(x)$ and rate $\lambda(x)$,
following the lead of differential stability \cite{Forni2014} and differential 
positivity \cite{Forni2016,Forni2015}.  Finally, the interconnection theorems in the paper
only consider quadratic supplies and storages but the concept of $p$-dissipativity is of
course more general.

\section*{Acknowledgment}
The authors wish to thank I. Cirillo and F. Miranda for 
useful comments and suggestions to the manuscript. 
\appendix
\begin{proofof}{\emph{Proposition \ref{prop:equivalences}}.}

[\emph{1) implies LMI \eqref{eq:LMI-dominance}}] Given the splitting of eigenvalues, 
by coordinate transformation $A+\lambda I $ can be expressed
in the block diagonal form 
$
{\small
A + \lambda I  = \mymatrix{c|c}{A_u+ \lambda I & 0  \\  \hline 0 & A_s+\lambda I} 
}
$,
where $A_u+\lambda I$ has $p$ unstable eigenvalues and $A_s+\lambda I$ has $n-p$ stable eigenvalues. Thus, there exist 
positive definite matrices $P_u$ and $P_s$ of rank $p$ and $n-p$ respectively
such that 
$
{\small
P := \mymatrix{c|c}{-P_u & 0  \\  \hline 0 & P_s} 
}
$
satisfies $(A+\lambda I)^T P + P (A+\lambda I) < 0$.
\eqref{eq:LMI-dominance} follows.

[\emph{LMI \eqref{eq:LMI-dominance} implies 1)}]
By coordinate transformation, without loss of generality, consider the block diagonal
representation 
$
{\small
A + \lambda I  = \mymatrix{c|c}{A_u  & 0  \\  \hline 0 & A_s } 
}
$
where $A_u$ is a $r\times r$ matrix whose eigenvalues have positive real part and 
of $A_s$ is a $q\times q$ matrix whose eigenvalues have non-positive real part. Clearly $r+q=n$.
In the same coordinates, consider $P = \mymatrix{c|c}{P_u  & \star  \\  \hline \star & P_s }$
and note that 
the LMI \eqref{eq:LMI-dominance} reads
$
{\small \mymatrix{c|c}{A_u  & 0  \\  \hline 0 & A_s }^{\!T} \!\!\mymatrix{c|c}{P_u  & \star  \\  \hline \star & P_s } 
+ \mymatrix{c|c}{P_u  & \star  \\  \hline \star & P_s }\!\mymatrix{c|c}{A_u  & 0  \\  \hline 0 & A_s } < 0} \vspace{1mm}
$
which entails
$A_u^T P_u + P_u A_u < 0$
and 
$A_s^T P_s + P_s A_s  < 0$.
The strict inequality of the latter guarantees that the eigenvalues of $A_s$ are strictly negative. 
Furthermore, necessarily,
$P_u$ has $r$ negative eigenvalues and $P_s$ has $q$ positive eigenvalues.
By the assumption on the inertia of $P$ and 
by \cite[Lemma 2]{Gilber1991}, $r \geq p$ and $q \geq n-p$. Since $n=r+q$, it follows that $r=p$ and $q=n-p$.

[1) $\Leftrightarrow$ 2)]. The equivalence between 1) and 2) follows from 
standard properties of linear systems, using the fact that 
the trajectories of $\dot{\bar{x}} = (A+\lambda I) \bar{x}$ and of
$\dot{x} = A x$ from the same initial condition $x(0) = \bar{x}(0)$ 
satisfy $x(t) = e^{-\lambda t} \bar{x}(t)$ for each $t\in \real$.
\end{proofof}

\bibliographystyle{plain}

\end{document}